\documentclass{eptcs}
\title{A Graph Theoretic Perspective on CPM(Rel)}
\author{Daniel Marsden}
\newcommand{\authorrunning}{Daniel Marsden}
\newcommand{\titlerunning}{A Graph Theoretic Perspective on CPM(Rel) }
\newcommand{\event}{QPL 2015}

\usepackage[theorems]{macros}
\usepackage{names}
\usepackage{eqproof}
\usepackage{extarrows}
\usepackage{amssymb}
\usepackage{amsmath}

\newcommand{\cpname}{\ensuremath{\text{\bf CPM}}}
\newcommand{\cp}[1]{\ensuremath{\cpname(#1)}}

\newcommand{\vertices}[1]{\ensuremath{V_{#1}}}
\newcommand{\edges}[1]{\ensuremath{E_{#1}}}
\newcommand{\graphcat}{\ensuremath{\mathcal{G}}\xspace}

\tikzset{->-/.style={decoration={
  markings,
  mark=at position .5 with {\arrow{>}}},postaction={decorate}}}

\tikzset{trapnode/.style={draw, trapezium, trapezium stretches=true, minimum width=1cm, minimum height=0.6cm}}
\tikzset{brnode/.style={trapnode, trapezium left angle=90, trapezium right angle=70}}
\tikzset{trnode/.style={trapnode, trapezium left angle=90, trapezium right angle=110}}
\tikzset{blnode/.style={trapnode, trapezium left angle=70, trapezium right angle=90}}
\tikzset{tlnode/.style={trapnode, trapezium left angle=110, trapezium right angle=90}}
\tikzset{stringdiagram/.style={scale=0.5, line width=0.25mm}}
\tikzset{blackdot/.style={circle,draw=black,fill=black,minimum size=1mm,inner sep=0mm}}
\tikzset{cpwire/.style={black, line width=0.5mm}}

\begin{document}

\maketitle

\begin{abstract}
Mixed states are of interest in quantum mechanics for modelling partial information.
More recently categorical approaches to linguistics have also exploited the idea
of mixed states to describe ambiguity and hyponym / hypernym relationships.
In both these application areas the category \rel of sets and binary relations is
often used as an alternative model.
Selinger's CPM construction provides the setting for mixed states in Hilbert space 
based categorical quantum mechanics. By analogy, applying the CPM construction to \rel is 
seen as introducing mixing into a relational setting.

We investigate the category \cp{\rel} of completely positive maps
in \rel. We show that the states of an object in \cp{\rel} are in
bijective correspondence with certain graphs. Via map-state duality
this then allows us provide a graph theoretic characterization of the morphisms
in \cp{\rel}. By identifying an appropriate composition operation on graphs, we
then show that \cp{\rel} is isomorphic to a category of sets and
graphs between them. This isomorphism then leads to a graph based description of
the complete join semilattice enriched $\dagger$-compact structure of \cp{\rel}. 
These results allow us to reason about \cp{\rel} entirely in terms of graphs. 

We exploit these techniques in several examples. We give a closed form expression
for the number of states of a finite set in \cp{\rel}. 
The pure states are characterized in graph theoretic terms.
We also demonstrate the possibly surprising phenomenon of a pure state that can be given as a mixture of two mixed states.
\end{abstract}

\section{Introduction}
We study aspects of Selinger's CPM construction \cite{Selinger2007}, and particularly
its application to the category \rel of sets and binary relations between them.
The CPM construction was originally introduced to extend the categorical approach to quantum mechanics
\cite{AbramskyCoecke2008} to support mixed states. The construction takes a $\dagger$-compact closed category $\mathcal{C}$
and produces a new $\dagger$-compact closed category \cp{\mathcal{C}}. 
In the case of quantum mechanics this construction takes us from the setting
of \cite{AbramskyCoecke2008}, modelling pure state quantum mechanics, to a setting suitable for mixed state quantum mechanics.

The distributional compositional semantics approach to linguistics \cite{CoeckeSadrzadehClark2010} uses
similar mathematical methods to categorical quantum mechanics. Density operators have been
proposed as candidates for modelling linguistic ambiguity \cite{Piedeleu2014, PiedeleuKartsaklisCoeckeSadrzadeh2015}
and hyponym / hypernym relationships \cite{Balkir2014}. In this way the CPM construction becomes
important in categorical models of linguistics.

In both linguistics and quantum mechanics it is common to consider \rel as an alternative model.
By analogy with the situation in quantum mechanics, the category \cp{\rel} is then seen as introducing
mixing in the relational setting. 
For example, a toy model of linguistic ambiguity in \cp{\rel} is discussed in 
\cite{PiedeleuKartsaklisCoeckeSadrzadeh2015}, with the set:
\begin{equation*}
\{\true, \false\}
\end{equation*}
considered as  representing truth values. 
In \rel this set has 4 states, corresponding to its subsets. In \cp{\rel} there are 5 states,
4 corresponding to the pure states, and a new mixed state. This raises several questions. Why 5 states?
How many states would we expect for an arbitrary finite set in \cp{\rel}? Can we easily identify the
pure states?
Direct computation gives the following state counts for small sets in \rel and \cp{\rel}:
\begin{center}
\begin{tabular}{l l l}
{\bf Elements } & {\bf \rel States} & {\bf \cp{\rel} States}\\
\hline
0 & 1 & 1 \\
1 & 2 & 2 \\
2 & 4 & 5 \\
3 & 8 & 18 \\
4 & 16 & 113 \\
5 & 32 & 1450
\end{tabular}
\end{center}
The numbers of states in \rel is simply the size of the corresponding
power set. The pattern for the \cp{\rel} states is more complex, understanding it
properly leads to a better understanding of the structure of \cp{\rel}.
It turns out that the states in \cp{\rel} correspond to certain graphs on the elements of the underlying sets. 
In this paper we:
\begin{itemize}
 \item Give a graph based classification of the states in \cp{\rel}.
 \item Via map-state duality give a classification of the morphisms in \cp{\rel}.
 \item By defining an appropriate graph composition operation, we exhibit a category
  of graphs isomorphic to \cp{\rel}. In this way we give a combinatorial description of
  \cp{\rel}.
 \item Via this isomorphism, we transfer the complete join semilattice enriched $\dagger$-compact 
  monoidal structure of \rel to our category of graphs, and give explicit graph based descriptions of this structure.
\end{itemize}
Along the way, we also characterize the pure states in \cp{\rel} in terms of graphs. We
then use this characterization to exhibit a pure state in \cp{\rel} that can be constructed
as a mixture of two mixed states. We also provide a closed form for the number of states of
a finite set in \cp{\rel}. 

We remark that although we must formalize our observations mathematically in later sections,
the spirit of our intention is that a simple observation improves our understanding of \cp{\rel}.
We provide explicit graphical examples throughout, and it is in these simple finite examples
that we anticipate the graphical reasoning to be of practical use. We also consider our
approach to be complementary to the usual string diagrammatic reasoning, rather than a competitor
to it. An application of the characterization of \cp{\rel} in terms of graphs, in a quantum mechanical setting,
can be found in \cite{Gogioso2015}.

\section{Preliminaries}
We assume some familiarity with $\dagger$-compact closed categories
and their graphical calculus \cite{AbramskyCoecke2004, Selinger2011, KellyLaplaza1980}.
This section provides some mathematical background on the key CPM construction, mainly from \cite{Selinger2007}.
\begin{definition}[Positive Morphism]
An endomorphism $f : A \rightarrow A$ in a $\dagger$-compact monoidal category is \define{positive}
if there exists an object $B$ and a morphism $g : A \rightarrow B$ with:
\begin{equation*}
f = g^\dagger \circ g
\end{equation*}
\end{definition}
\begin{definition}[Completely Positive Morphisms]
For a $\dagger$-compact monoidal category $\mathcal{C}$, the category of \define{completely positive morphisms} 
, \cp{\mathcal{C}}, has:
\begin{itemize}
 \item {\bf Objects}: The same objects as $\mathcal{C}$
 \item {\bf Morphisms}: Let $X^{*}$ denote the dual of an object $X$. 
  A morphism $f : A \rightarrow B$ is a $\mathcal{C}$ morphism $f : A \otimes A^{*} \rightarrow B \otimes B^{*}$
  such that the following composite is positive:
\begin{equation*}
\begin{tikzpicture}[stringdiagram, scale=0.5]
\path node[blnode] (f) {$f$}
 (f.south) ++(-1,0) coordinate (a) ++(-1,-1) coordinate (b) ++(-1,1) coordinate (c)
 (f.south) ++(1,0) coordinate (d) ++(0,-1.5) coordinate[label=below:$A^*$] (bl)
 (f.north) ++(1,0) coordinate (e) ++(1,1) coordinate (g) ++(1,-1) coordinate (h)
 (f.north) ++(-1,0) coordinate (i) ++(0,1.5) coordinate[label=above:$B$] (tr);
\path let \p1 = (tr) in
      let \p2 = (bl) in
      let \p3 = (c) in
      let \p4 = (h) in
 coordinate[label=below:$B$] (br) at (\x4, \y2)
 coordinate[label=above:$A^*$] (tl) at (\x3, \y1);
\draw (a) to[out=-90, in=0] (b.west) -- (b.east) to[out=180, in=-90] (c) -- (tl)
 (d) -- (bl)
 (e) to[out=90, in=180] (g.west) -- (g.east) to[out=0, in=90] (h) -- (br)
 (i) -- (tr);
\end{tikzpicture}
\end{equation*}
\end{itemize}
Composition and identities are inherited from $\mathcal{C}$. 
\end{definition}
The underlying category embeds into the result of the CPM construction \cite{Selinger2007}:
\begin{lemma}
\label{lem:embedding}
If $\mathcal{C}$ is a compact closed category there is a surjective and identity on objects functor:
\begin{align*}
\mathcal{C} &\rightarrow \cp{\mathcal{C}}\\
\begin{gathered}
\begin{tikzpicture}[stringdiagram]
\path node[blnode] (f) {$f$}
 +(0,-1) coordinate[label=below:$A$] (b)
 +(0,1) coordinate[label=above:$B$] (t);
\draw (b) -- (f.south)
 (f.north) -- (t);
\end{tikzpicture}
\end{gathered}
&\mapsto
\begin{gathered}
\begin{tikzpicture}[stringdiagram]
\path node[blnode] (f) {$f$}
 +(0,-1) coordinate[label=below:$A$] (bl)
 +(0,1) coordinate[label=above:$B$] (tl)
 ++(2,0) node[brnode] (g) {$f$}
 +(0,-1) coordinate[label=below:$A^*$] (br)
 +(0,1) coordinate[label=above:$B^*$] (tr);
\draw (bl) -- (f.south)
 (f.north) -- (tl);
\draw (br) -- (g.south)
 (g.north) -- (tr);
\end{tikzpicture}
\end{gathered}
\end{align*}
\end{lemma}
\begin{definition}
  \label{def:self-dual}
  We say that a compact closed category is \define{self-dual} if each object is equal to its dual.
\end{definition}
\begin{remark}
  Definition \ref{def:self-dual} is a naive definition of self-duality, but sufficient for the purposes of this paper.
  More careful analysis of self-duality, including a broader definition and suitable coherence conditions is given in \cite{Selinger2010}.
\end{remark}
The following theorem is a small variation on \cite[theorem 4.2]{Selinger2007}. By slightly changing the definition of the tensor product
used in Selinger's work, we show that \emph{self-dual} compact structure on $\mathcal{C}$ lifts to \emph{self-dual} compact structure on \cp{\mathcal{C}}.
We provide this variation as the symmetrical form of the tensor product is slightly more convenient to work with for self-dual compact closed categories such as \rel.
\begin{theorem}
\label{thm:selfdual}
If $\mathcal{C}$ is a $\dagger$-compact closed category with self-dual $\dagger$-compact
structure then \cp{\mathcal{C}} is also a $\dagger$-compact closed category with
self-dual $\dagger$-compact structure.
\end{theorem}
\begin{proof}
We must define the tensor product to account for self-duality in $\mathcal{C}$.
The action on objects is inherited from the tensor in $\mathcal{C}$.
For arbitrary \cp{\mathcal{C}}-morphisms $f: A \rightarrow B$ and $g: B \rightarrow D$, the action of the tensor product bifunctor
on morphisms is given by the following composite in $\mathcal{C}$:
\begin{equation*}
\begin{gathered}
\begin{tikzpicture}[stringdiagram]
\path node[blnode] (f) {$f$}
 (f.north) ++(-0.4,0) coordinate (ftl) ++(0,1) coordinate[label=above:$C$] (tll)
 (f.south) ++(-0.4,0) coordinate (fbl) ++(0,-1) coordinate[label=below:$A$] (bll)
 (f.north) ++(0.4,0) coordinate (ftr) ++(0,1) coordinate[label=above:$D$] (tl)
 (f.south) ++(0.4,0) coordinate (fbr) ++(0,-1) coordinate[label=below:$B$] (bl)
 (f) ++(3,0) node[blnode] (g) {$g$}
 (g.north) ++(-0.4,0) coordinate (gtl) ++(0,1) coordinate[label=above:$C$] (tr)
 (g.south) ++(-0.4,0) coordinate (gbl) ++(0,-1) coordinate[label=below:$A$] (br)
 (g.north) ++(0.4,0) coordinate (gtr) ++(0,1) coordinate[label=above:$D$] (trr)
 (g.south) ++(0.4,0) coordinate (gbr) ++(0,-1) coordinate[label=below:$B$] (brr);

\draw (bll) -- (fbl)
      (ftl) -- (tll)
      (brr) -- (gbr)
      (gtr) -- (trr);
\draw (ftr) to[out=90, in=-90] (tr)
      (gtl) to[out=90, in=-90] (tl)
      (fbr) to[out=-90, in=90] (br)
      (gbl) to[out=-90, in=90] (bl);
\end{tikzpicture}
\end{gathered}
\end{equation*}
\cp{\mathcal{C}} morphisms are closed under this operation, and
bifunctoriality is easy to see.
The monoidal unit and associators and unitors, along with the $\dagger$-compact
structure are inherited from $\mathcal{C}$.
\end{proof}
As \rel has self-dual $\dagger$-compact structure, we shall assume the monoidal structure given in theorem \ref{thm:selfdual} in later sections.
We will require the following observation from \cite{Selinger2007}:
\begin{lemma}
\label{lem:posrel}
A relation $R : X \rightarrow X$ is a positive morphism in \rel if and only if it
is symmetric and:
\begin{equation*}
R(x,x') \,\Rightarrow\, R(x,x)
\end{equation*}
\end{lemma}
This can be generalized to \cp{\rel} morphisms via map-state duality:
\begin{lemma}
\label{lem:cpmap}
The morphisms in \cp{\rel} are exactly the relations $R : A \otimes A \rightarrow B \otimes B$
satisfying the following axioms:
\begin{align*}
R(a_1, a_2, b_1, b_2) \quad&\Rightarrow\quad R(a_2, a_1, b_2, b_1)\\
R(a_1, a_2, b_1, b_2) \quad&\Rightarrow\quad R(a_1, a_1, b_1, b_1)
\end{align*}
\end{lemma}
\begin{proof}
This follows from map-state duality and lemma \ref{lem:posrel}.
For an arbitrary relation $R : A \otimes A \rightarrow B \otimes B$, we define relation:
\begin{align*}
\overline{R} : A \otimes B &\rightarrow A \otimes B\\
\overline{R}(a_1, b_1, a_2, b_2) &\Leftrightarrow R(a_2, a_1, b_2, b_1)
\end{align*}
$\overline{R}$ is symmetric iff:
\begin{equation*}
R(a_1, a_2, b_1, b_2) \Rightarrow R(a_2, a_1, b_2, b_1)
\end{equation*}
$\overline{R}$ satisfies:
\begin{equation*}
\overline{R}(a_1, b_1, a_2, b_2) \Rightarrow \overline{R}(a_1, b_1, a_1, b_1)
\end{equation*}
iff:
\begin{equation*}
R(a_2, a_1, b_2, b_1) \Rightarrow R(a_1, a_1, b_1, b_1)
\end{equation*}
\end{proof}

\section{\cp{\rel} States as Graphs}
We now start to relate \cp{\rel} to suitable types of graphs.
\begin{definition}
An undirected graph $\gamma$ is:
\begin{itemize}
 \item \define{Labelled} if the vertices can be distinguished.
 \item \define{Simple} if every vertex has a self-loop and there are no duplicate edges.
\end{itemize}
A \define{complete graph} is a graph in which there is an edge between
each pair of vertices. We will write:
\begin{itemize}
 \item \vertices{\gamma} for the vertex set of $\gamma$.
 \item \edges{\gamma} for the edge set of $\gamma$.
\end{itemize}
In the sequel, we will refer to simple undirected labelled graphs simply as graphs.
\end{definition}
\begin{remark}
In the literature, simple graphs are normally required to have \emph{no} self-loops. Such
graphs clearly bijectively correspond to those of our definition. The presence of self-loops
will render some later definitions and reasoning more uniform. 
When drawing diagrams of graphs we will leave self-loops implicit to avoid clutter.
\end{remark}
The following observation is at the heart of everything that follows:
\begin{proposition}
\label{prop:graphchar}
The states of a set $X$ in \cp{\rel} bijectively correspond to the graphs on subsets of elements of $X$.
\end{proposition}
\begin{proof}
Given a point of $X$ in \cp{\rel}, we construct a graph with vertices:
\begin{equation*}
\{ x \mid R(x,x) \}
\end{equation*}
We have an edge from $x$ to $x'$ iff:
\begin{equation*}
R(x,x')
\end{equation*}
For a graph on a subset of elements of $X$, we define a positive morphism $R$
in \rel with $R(x,x')$ iff there is an edge from $x$ to $x'$. This is a positive morphism
as our graphs are undirected and have self-loops on all vertices.

These two constructions are inverse to each other, and therefore witness the required bijection.
\end{proof}
We are now able to provide a closed form for the number of states of a finite set in \cp{\rel}.
\begin{corollary}
An $n$-element finite set in \cp{\rel} has:
\begin{equation*}
\sum_{0 \leq i \leq n} \binom{n}{i} 2^{i (i - 1) /2}
\end{equation*}
states.
\end{corollary}
\begin{proof}
There are $2^{i(i-1)/2}$ graphs with $i$ vertices. There are $\binom{n}{i}$
ways of choosing $0 \leq i \leq n$ vertices from $X$. We then simply count all these possible graphs to get
our total.
\end{proof}
We now characterize pure states in a manner that is literally easy to see:
\begin{proposition}
\label{prop:graphmix}
A state in \cp{\rel} is pure iff the corresponding graph is complete.
\end{proposition}
\begin{proof}
We identify subsets $U \subseteq X$ with relations $U : \{*\} \rightarrow X$ in the obvious way.

For $U \subseteq X$, the positive morphism corresponding to this pure state
is given by $U \circ U^{\circ}$. Now assume distinct $u,v \in U$. We then have $U^{\circ}(u,*) \wedge U(*,v)$
and so $U \circ U^{\circ}(u,v)$. The corresponding graph of $U$ therefore has an edge between each pair
of distinct elements of $U$.

In the opposite direction, a connected graph on vertex set $V$ corresponds to the positive morphism given by
the relation $V \circ V^\circ$.
\end{proof}
We note the following observation that runs counter to some of our intuition about how mixing should
behave:
\begin{lemma}
\label{lem:mixing}
There exist pure states in \cp{\rel} that can be given as a ``convex combination'' (union) of mixed states.
\end{lemma}
\begin{proof}
We consider the following two states of the set $\{x,y,z\}$:
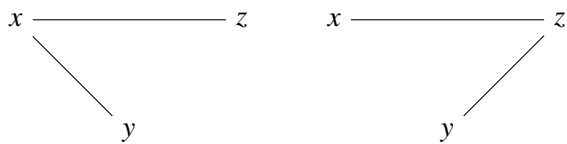
\begin{equation*}
\begin{gathered}
\begin{tikzpicture}[scale=0.5,node distance=1.5cm]
\node (tl) {$x$};
\node[below of=tl, right of=tl] (bot) {$y$};
\node[above of=bot, right of=bot] (tr) {$z$};
\draw (tl) -- (tr) (tl) -- (bot);
\end{tikzpicture}
\end{gathered}
\qquad
\begin{gathered}
\begin{tikzpicture}[scale=0.5, node distance=1.5cm]
\node (tl) {$x$};
\node[below of=tl, right of=tl] (bot) {$y$};
\node[above of=bot, right of=bot] (tr) {$z$};
\draw (tl) -- (tr) (tr) -- (bot);
\end{tikzpicture}
\end{gathered}
\end{equation*}
We note that these are not pure states by proposition \ref{prop:graphmix}. If we form the union of the underlying
relations, we get the state:
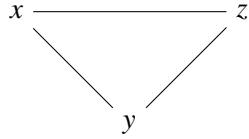
\begin{equation*}
\begin{tikzpicture}[scale=0.5, node distance=1.5cm]
\node (tl) {$x$};
\node[below of=tl, right of=tl] (bot) {$y$};
\node[above of=bot, right of=bot] (tr) {$z$};
\draw (tl) -- (tr) (tl) -- (bot) (tr) -- (bot);
\end{tikzpicture}
\end{equation*}
As this graph is complete, the corresponding state is a pure state by proposition \ref{prop:graphmix}.
\end{proof}

\section{A Category of Graphs}
We now pursue a description of the category \cp{\rel} entirely in terms graphs and operations upon graphs.
Via map-state duality, as a corollary of proposition \ref{prop:graphchar}, there is a bijection between
$\cp{\rel}(A,B)$ and the set of graphs with vertices a subset of $A \times B$. We will therefore
be interested in graphs of this form when considering the morphisms in \cp{\rel}.
\begin{definition}
\label{def:graphon}
For sets $A,B$:
\begin{itemize}
 \item If $\gamma$ is a graph with vertex labels a subset of $A \times B$, we will write $\gamma : A \rightarrow B$.
 \item For a set $A$ we define the graph $1_A : A \rightarrow A$ as the complete graph with vertices the diagonal of $A \times A$.
\end{itemize}
\end{definition}

\begin{definition}[Graph Composition]
\label{def:graphcomp}
For graphs $\gamma : A \rightarrow B$ and $\gamma' : B \rightarrow C$,
we define the graph $\gamma' \circ \gamma : A \rightarrow C$ as follows:
\begin{align*}
\vertices{\gamma'\circ \gamma} &:= \{ (a,c) \mid \exists b \in B. \; (a,b) \in \vertices{\gamma} \;\wedge (b,c) \; \in \vertices{\gamma'} \}\\
\edges{\gamma' \circ \gamma} &:= \{ \{ (a,c), (a',c') \} \mid \exists b, b' \in B. \; \{(a,b),(a',b')\}\in \edges{\gamma} \; \wedge\; \{(b,c),(b',c')\} \in \edges{\gamma'} \}
\end{align*}
\end{definition}
\begin{remark}
Informally, we glue together vertices and edges where their $B$ components agree. Notice
that in order to calculate the vertices of the composite, we only need consider vertices
in the components. Similarly, when calculating the edges, we need only consider edges in 
the components.
\end{remark}

\begin{example}
The composite of the following two graphs:
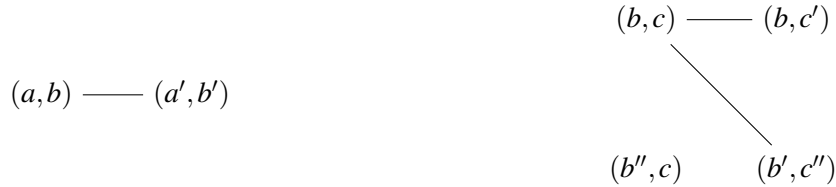
\begin{trivlist}\item
\begin{minipage}{0.495\textwidth}
\begin{center}
\begin{tikzpicture}[node distance=2cm]
\node (l) {$(a,b)$};
\node[right of=l] (r) {$(a',b')$};
\draw (l) -- (r);
\end{tikzpicture}
\end{center}
\end{minipage}
\begin{minipage}{0.495\textwidth}
\begin{center}
\begin{tikzpicture}[node distance=2cm]
\node (tl) {$(b,c)$};
\node[below of=tl] (bl) {$(b'',c)$};
\node[right of=tl] (tr) {$(b,c')$};
\node[below of=tr] (br) {$(b',c'')$};
\draw (tl) -- (tr)
      (tl) -- (br);
\end{tikzpicture}
\end{center}
\end{minipage}
\end{trivlist}
is given by the graph:
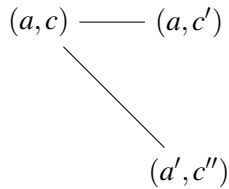
\begin{center}
\begin{tikzpicture}[node distance=2cm]
\node (l) {$(a,c)$};
\node[right of=l] (tr) {$(a,c')$};
\node[below of=tr] (br) {$(a',c'')$};
\draw (l) -- (tr) 
      (l) -- (br);
\end{tikzpicture}
\end{center}
\end{example}

\begin{lemma}
\label{lem:compwellbehaved}
For graphs $\gamma : A \rightarrow B$, $\gamma' : B \rightarrow C$, $\gamma'' : C \rightarrow D$,
and $1_{A}, 1_B$ as in definition \ref{def:graphon}:
\begin{itemize}
 \item $\gamma'' \circ (\gamma' \circ \gamma) = (\gamma'' \circ \gamma') \circ \gamma$
 \item $1_B \circ \gamma = \gamma$
 \item $\gamma \circ 1_A = \gamma$
\end{itemize}
\end{lemma}

\begin{definition}
We define the category \graphcat of sets and graphs between them as having:
\begin{itemize}
 \item {\bf Objects}: Sets
 \item {\bf Morphisms}: A morphism $A \rightarrow B$ is a graph $\gamma : A \rightarrow B$ as in definition \ref{def:graphon}
\end{itemize}
Composition is as described in definition \ref{def:graphcomp}. The identities and associativity
of composition then follow from lemma \ref{lem:compwellbehaved}
\end{definition}
We will require the following technical lemma:
\begin{lemma}
\label{lem:tech}
Let $R : A \rightarrow B$ and $S : B \rightarrow C$ be morphisms in \cp{\rel}. Then for all $a \in A$
and $c \in C$:
\begin{equation*}
(S \circ R)(a,a,c,c)\quad\Leftrightarrow\quad \exists b \in B.\,R(a,a,b,b)\,\wedge\,S(b,b,c,c)
\end{equation*}
\end{lemma}
\begin{proof}
The right to left direction is immediate from the definition of relational composition. For the other
direction, if:
\begin{equation*}
(S \circ R)(a,a,c,c)
\end{equation*}
there then exist $b,b' \in B$ such that:
\begin{equation*}
R(a,a,b,b') \quad\wedge\quad S(b,b',c,c)
\end{equation*}
By lemma \ref{lem:cpmap} we then have:
\begin{equation*}
R(a,a,b,b) \quad\wedge\quad S(b,b,c,c)
\end{equation*}
\end{proof}

\begin{definition}
\label{def:graphfunctor}
For a \cp{\rel} morphism $R : A \rightarrow B$, we define the graph $G(R) : A \rightarrow B$:
\begin{align*}
\vertices{G(R)} &:= \{ (a,b) \mid R(a,a,b,b) \}\\
\edges{G(R)} &:= \{ \{ (a,b), (a',b') \} \mid R(a,a',b,b') \}
\end{align*}
\end{definition}

\begin{proposition}
The mapping of definition \ref{def:graphfunctor} extends to
an identity on objects functor:
\begin{equation*}
G : \cp{\rel} \rightarrow \graphcat
\end{equation*}
\end{proposition}
\begin{proof}
We first check identities are preserved:
\begin{trivlist}\item
\begin{minipage}{0.495\textwidth}
\begin{eqproof*}
\vertices{G(1_A)}
\explain{ definition \eqref{def:graphfunctor} }
\{ (a_1, a_2) \mid R(a_1, a_1, a_2, a_2) \}
\explain{ identity relations }
\{ (a_1, a_2) \} 
\explain{ definition }
\vertices{1_{G(A)}}
\end{eqproof*}
\end{minipage}
\begin{minipage}{0.495\textwidth}
\begin{eqproof*}
\edges{G(1_A)}
\explain{ definition \eqref{def:graphfunctor} }
\{ \{ (a_1, a_2),(a_3, a_4) \} \mid 1_A(a_1, a_3, a_2, a_4) \}
\explain{ identity relations }
\{ \{ (a_1, a_1),(a_3, a_3) \}
\explain{ definition }
\edges{1_{G(A)}}
\end{eqproof*}
\end{minipage}
\end{trivlist}
We must also check composition is preserved. For vertices:
\begin{eqproof*}
\vertices{G(S \circ R)}
\explain{ definition \ref{def:graphfunctor} }
\{ (a,c) \mid (S \circ R)(a,a,c,c) \}
\explain{ lemma \ref{lem:tech} }
\{ (a,c) \mid \exists b.\; R(a,a,b,b) \;\wedge\; S(b,b,c,c) \}
\explain{ definition \ref{def:graphfunctor} }
\{ (a,c) \mid \exists b.\;(a,b) \in \vertices{G(R)}\;\wedge\; (b,c) \in \vertices{G(S)} \}
\explain{ definition \ref{def:graphcomp} }
\vertices{G(S) \circ G(R)}
\end{eqproof*}
For edges:
\begin{eqproof*}
\edges{G(S \circ R)}
\explain{ definition \ref{def:graphfunctor} }
\{ \{ (a,c), (a',c')\} \mid (S \circ R)(a,a',c,c') \}
\explain{ relation composition }
\{ \{ (a,c), (a',c')\} \mid \exists b,b' \in B.\;R(a,a',b,b') \;\wedge\; S(b,b',c,c') \}
\explain{ definition \ref{def:graphfunctor} }
\{ \{(a,c), (a',c')\} \mid \exists b,b' \in B.\;\{ (a,b), (a',b') \} \in \edges{G(R)}\;\wedge\;\{ (b,c), (b',c') \} \in \edges{G(S)} \}
\explain{ definition \ref{def:graphcomp} }
\edges{G(S) \circ G(R) }
\end{eqproof*}
\end{proof}

\begin{definition}
\label{def:relfunctor}
For a graph $\gamma : A \rightarrow B$ we define the \cp{\rel} morphism $C(\gamma) : A \rightarrow B$
as follows:
\begin{equation*}
C(\gamma)(a,a',b,b') \Leftrightarrow \{ (a,b), (a',b') \} \in \edges{\gamma}
\end{equation*}
\end{definition}

\begin{proposition}
The mapping  of definition \ref{def:relfunctor} extends to an
identity on objects functor:
\begin{equation*}
C : \graphcat \rightarrow \cp{\rel}
\end{equation*}
\end{proposition}
\begin{proof}
That the mapping is identity on objects is trivial. To confirm identities are preserved, we have:
\begin{eqproof*}
C(1_A)(a_1,a_2,a_3,a_4)
\explain[\Leftrightarrow]{ definition \eqref{def:relfunctor} }
\{ (a_1, a_3), (a_2, a_4) \} \in \edges{1_A}
\explain[\Leftrightarrow]{ definition }
a_1 = a_3\;\wedge\;a_2 = a_4
\explain[\Leftrightarrow]{ identity relations }
1_{C(A)}(a_1,a_2,a_3,a_4)
\end{eqproof*}
For associativity of composition, we reason as follows:
\begin{eqproof*}
C(\gamma' \circ \gamma)(a,a',c,c')
\explain[\Leftrightarrow]{ definition \ref{def:relfunctor} }
\{ (a,c),(a',c') \} \in \edges{\gamma' \circ \gamma }
\explain[\Leftrightarrow]{ definition \ref{def:graphcomp} }
\exists b,b' \in B.\;\{ (a,b),(a',b') \} \in \edges{\gamma}\;\wedge\;\{ (b,c), (b',c') \} \in \edges{\gamma'}
\explain[\Leftrightarrow]{ definition \ref{def:relfunctor} }
\exists b,b' \in B.\;C(\gamma)(a,a',b,b')\;\wedge\;C(\gamma') (b,b',c,c')
\explain[\Leftrightarrow]{ relation composition }
(C(\gamma') \circ C(\gamma))(a,a',c,c')
\end{eqproof*}
\end{proof}

\begin{theorem}
\label{theorem:catiso}
The functors $G : \cp{\rel} \rightarrow \graphcat$ of definition \ref{def:graphfunctor} and $C : \graphcat \rightarrow \cp{\rel}$ 
of definition \ref{def:relfunctor} are inverse to each other.
\end{theorem}
\begin{proof}
In the first direction:
\begin{eqproof*}
(C \circ G)(R)(a,a',b,b')
\explain[\Leftrightarrow]{ definition \ref{def:relfunctor} }
\{ (a,b), (a',b') \} \in \edges{G(R)}
\explain[\Leftrightarrow]{ definition \ref{def:graphfunctor} }
R(a,a',b,b')
\end{eqproof*}
For the other direction, for vertices:
\begin{eqproof*}
\vertices{(G \circ C)(\gamma)}
\explain{ definition \ref{def:graphfunctor} }
\{ (a,b) \mid C(\gamma)(a,a,b,b) \}
\explain{ definition \ref{def:relfunctor} }
\{ (a,b) \mid \{ (a,b) \} \in \edges{\gamma} \}
\explain{ self-loops correspond to vertices }
\{ (a,b) \mid (a,b) \in \vertices{\gamma} \}
\explain{ set theory }
\vertices{\gamma}
\end{eqproof*}
Finally, for edges:
\begin{eqproof*}
\edges{(G \circ G)(\gamma)}
\explain{ definition \ref{def:graphfunctor} }
\{ \{ (a,b), (a',b') \} \mid C(\gamma)(a,a',b,b') \}
\explain{ definition \ref{def:relfunctor} }
\{ \{ (a,b), (a',b') \} \mid \{ (a,b), (a',b') \} \in \edges{\gamma} \}
\explain{ set theory }
\edges{\gamma}
\end{eqproof*}
\end{proof}

\section{A Structured Category of Graphs}
As the categories \graphcat and \cp{\rel} are isomorphic, we can transfer structure from \cp{\rel} to \graphcat. 
Firstly, the monoidal structure on \cp{\rel} induces a monoidal structure on \graphcat. 
As the functors witnessing the isomorphism are identity on objects, the tensor of objects in
\graphcat is also given by the cartesian product of sets.
The tensor product of morphisms in \graphcat is given explicitly as follows:
\begin{definition}[Graph Tensor]
\label{def:monoidal}
For graphs $\gamma : A \rightarrow C$, $\gamma' : B \rightarrow D$
we define the graph $\gamma \otimes \gamma' : A \otimes B \rightarrow C \otimes D$ as follows:
\begin{align*}
\vertices{\gamma \otimes \gamma'} &:= \{ (a,b,c,d) \mid (a,c) \in \vertices{\gamma} \; \wedge \; (b,d) \in \vertices{\gamma'}\}\\
\edges{\gamma \otimes \gamma'} &:= \{ \{ (a,b,c,d), (a',b',c',d') \} \mid \{ (a,c), (a',c') \} \in \edges{\gamma} \; \wedge \; \{ (b,d), (b',d') \} \in \edges{\gamma'} \} 
\end{align*}
\end{definition}
\begin{remark}
Informally, the vertices of the composite graph are all combinations of vertices from the two component graphs.
There is then an edge between a pair of vertices if there are edges between both pairs of component vertices.
Again, when calculating the vertices, we need only consider the vertices of the components, and when calculating
edges we need only consider component edges.
\end{remark}

\begin{example}
The tensor of the following pair of graphs:
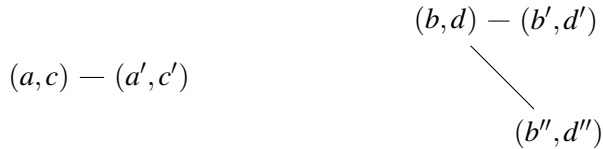
\begin{trivlist}\item
\begin{minipage}{0.495\textwidth}
\begin{center}
\begin{tikzpicture}[scale=0.5]
\path node(l) {$(a,c)$} ++(3,0) node(r) {$(a',c')$};
\draw (l) -- (r);
\end{tikzpicture}
\end{center}
\end{minipage}
\begin{minipage}{0.495\textwidth}
\begin{tikzpicture}[scale=0.5]
\path node(l) {$(b,d)$} ++(3,0) node(tr) {$(b',d')$} ++(0,-3) node(br) {$(b'',d'')$};
\draw (l) -- (tr) (l) -- (br);
\end{tikzpicture}
\end{minipage}
\end{trivlist}
is given by the graph:
\begin{center}
\begin{tikzpicture}[node distance=3cm]
\node (tl) {$(a,b'',c,d'')$};
\node[below of=tl] (ml) {$(a,b,c,d)$};
\node[below of=ml] (bl) {$(a,b',c,d')$};
\node[right of=tl] (tr) {$(a',b'',c',d'')$};
\node[below of=tr] (mr) {$(a',b,c',d)$};
\node[below of=mr] (br) {$(a',b',c',d')$};
\draw (tl) -- (tr)
      (ml) -- (mr)
      (bl) -- (br)
      (tl) -- (mr)
      (ml) -- (tr)
      (ml) -- (br)
      (bl) -- (mr)
      (tl) -- (ml)
      (ml) -- (bl)
      (tr) -- (mr)
      (mr) -- (br);
\end{tikzpicture}
\end{center}
\end{example}
\begin{definition}
\cjslat is the category of complete join semilattices and homomorphisms.
\end{definition}
The category \cp{\rel} is \cjslat-enriched, via the usual inclusion order on relations. Again,
as \cp{\rel} is isomorphic to \graphcat, we can transfer this enrichment to the homsets of
\graphcat. Concretely this ordering is given as follows:
\begin{definition}
\label{def:posenrichment}
We define a partial order on graphs:
\begin{equation*}
\gamma \leq \gamma' \quad\Leftrightarrow\quad \edges{\gamma} \subseteq \edges{\gamma'}
\end{equation*}
These graphs are closed under unions in the following sense:
\begin{equation*}
\vertices{\bigcup_i \gamma_i} := \bigcup_i \vertices{\gamma_i} \qquad\qquad\qquad \edges{\bigcup_i \gamma_i} := \bigcup_i \edges{\gamma_i}
\end{equation*}
\end{definition}
We can also recover a dagger operation on \graphcat, induced from the isomorphism with \cp{\rel}.
The dagger is given by ``reversing the order of the vertex pairs'', explicitly:
\begin{definition}[Graph Dagger]
\label{def:dagger}
We define identity on objects involution $(-)^\dagger : \graphcat^{op} \rightarrow \graphcat$.
The action on morphisms for graph $\gamma : A \rightarrow B$ is given by:
\begin{align*}
\vertices{\gamma^\dagger} &:= \{ (b,a) \mid (a,b) \in \vertices{\gamma} \}\\
\edges{\gamma^\dagger} &:= \{ \{(b,a),(b',a')\} \mid \{ (a,b), (a',b') \} \in \edges{\gamma} \}
\end{align*}
\end{definition}
Via our isomorphism, we describe a direct, graph based version of the embedding of lemma \ref{lem:embedding},
generalizing proposition \ref{prop:graphmix}:
\begin{proposition}
\label{prop:graphembedding}
There is a faithful, surjective and identity on objects functor:
\begin{equation*}
\rel \rightarrow \graphcat
\end{equation*}
The action on morphisms is to send a relation $R : A \rightarrow B$ to the complete graph on
$R \subseteq A \times B$. 
\end{proposition}
Using proposition \ref{prop:graphembedding}, we can easily transfer symmetric and compact structure to \graphcat. 
Concretely, this is given as follows:
\begin{definition}
\label{def:symcom}
The symmetry $A \otimes B \rightarrow B \otimes A$ is given by
the complete graph on:
\begin{equation*}
\{ (a,b,b,a) \mid a \in A, b \in B \}
\end{equation*}
The cup of the compact structure corresponds to the complete graph on:
\begin{equation*}
\{ (*,a,a) \mid a \in A \}
\end{equation*}
The cap is then given by applying the $\dagger$ graph operation.
Clearly the vertex sets described above are isomorphic to simpler sets. The definitions above will interact
correctly with our definitions of composition and tensor product.
\end{definition}
We can now extend theorem \ref{theorem:catiso} to account for the order-enriched $\dagger$-compact structure:
\begin{theorem}
Consider \graphcat as a \cjslat-enriched $\dagger$-compact closed category, with the required structure
as in definitions \ref{def:monoidal}, \ref{def:posenrichment}, \ref{def:dagger} and \ref{def:symcom}.
The categories \cp{\rel} and \graphcat are isomorphic as \cjslat-enriched $\dagger$-compact categories.
\end{theorem}

\section{Conclusion}
From the observations in this paper, we can view \cp{\rel} as a $\dagger$-compact monoidal category with
a highly combinatorial flavour. 

We already encounter some strange behaviour in \cp{\rel}, such as
the unusual mixing behaviour noted in lemma \ref{lem:mixing}.
One direction for further investigation is to explore the CPM construction for other more general categories
of relations. Every regular category has a corresponding category of relations, giving a broad
class of examples. Another generalized class of categories of relations is given by matrix algebras over
commutative quantales. Study of these examples may clarify the intuitive interpretation of the CPM construction
as introducing mixing, and furnish a broad class of models for string diagrammatic calculi.

A second area of interest for further investigation is iteration of the CPM construction. This is of practical
interest in linguistics where multiple aspects of mixing occur, such as ambiguity and hyponym / hypernym relations.
Here again, categories of relations might prove useful test cases for our intuitions as to how ``multiple levels
of mixing'' should behave.

\subsection*{Acknowledgements}
I would like to thank Martha Lewis, Bob Coecke, Stefano Gogioso and Peter Selinger for discussions
and helpful feedback. 

\bibliography{cpstates}
\bibliographystyle{eptcs}

\end{document}